\documentclass{llncs}
\usepackage{amsmath}
\usepackage{amssymb}
\usepackage{enumerate}
\usepackage{varioref}
\usepackage{charter,eulervm}
\usepackage{graphicx}
\usepackage{subfig,wrapfig}
\usepackage{boxedminipage}

\newcommand{\es}{\varnothing}

\title{{\sc On Independence Domination}}

\author{
 Wing-Kai~Hon\inst{1}
\and
 Ton~Kloks\inst{1}
\and 
 Hsiang~Hsuan~Liu\inst{1}
\and 
 Sheung-Hung~Poon\inst{1} 
\and
 Yue-Li~Wang\inst{2}} 
\institute{
 Department of Computer Science\\
 National Tsing Hua University, Taiwan\\
 {\tt \{wkhon,hhliu,kloks,spoon\}@cs.nthu.edu.tw} 
\and 
 Department of Information Management\\ 
 National Taiwan University of Science and Technology\\ 
 {\tt ylwang@cs.ntust.edu.tw}
}

\begin{document}

\maketitle

\begin{abstract}
Let $G$ be a graph. 
The independence-domination number $\gamma^i(G)$ is the 
maximum over all independent sets $I$ in $G$ of the 
minimal number of vertices needed to dominate $I$. 
In this paper we investigate the computational complexity 
of $\gamma^i(G)$ for graphs in several graph classes related 
to cographs. We present an exact exponential 
algorithm. We also present a PTAS for planar graphs. 
\end{abstract}

\section{Introduction}
\label{section intro}

Let $G=(V,E)$ be a graph. A set $A$ of vertices dominates a 
set $B$ if 
\[B \subseteq \bigcup_{x \in A}\; N[x].\] 
The minimal cardinality of a set of vertices needed 
to dominate a set $B$ is denoted by $\gamma_G(B)$. 
The domination number $\gamma(G)$ of the graph $G$ 
is thus defined as 
$\gamma_G(V)$, where $V$ is the set of vertices of $G$.  
When the graph $G$ is clear from the context we omit the 
subscript $G$. 

\bigskip 

\begin{definition}
The independence-domination number $\gamma^i(G)$ is 
\[\gamma^i(G)= \max \; \{\;\gamma(A)\;|\; \text{$A$ is an 
independent set in $G$}\;\}.\] 
\end{definition}

\bigskip 

Obviously, $\gamma(G) \geq \gamma^i(G)$. In~\cite{kn:aharoni} 
it was shown that $\gamma(G)=\gamma^i(G)$ for chordal graphs. 
Using this result Aharoni and Szab\'o showed that Vizing's 
conjecture on the domination number of the Cartesian product 
of graphs is true for chordal graphs, ie, 
\[\gamma(G \Box H) \geq \gamma(G) \cdot \gamma(H) 
\quad \text{when $G$ and $H$ are chordal~\cite{kn:aharoni2}.}\]  

\bigskip 

The Cartesian product $G \Box H$ is the graph which has pairs 
$(g,h)$, $g \in V(G)$ and $h \in V(H)$ as its vertices. Two pairs 
$(g_1,h_1)$ and $(g_2,h_2)$ are adjacent in $G \Box H$ if 
either $g_1=g_2$ and $\{h_1,h_2\} \in E(H)$ or 
$\{g_1,g_2\} \in E(G)$ and $h_1=h_2$~\cite{kn:imrich}. 
Vizing conjectured in 1968~\cite{kn:vizing} that 
\[\text{for all graphs $G$ and $H$}\quad
\gamma(G \Box H) \geq \gamma(G) \cdot \gamma(H).\] 
In 1994 Fisher proved that 
\begin{equation}
\label{eqn0}
\text{for all connected graphs $G$ and $H$}\quad 
\gamma(G \Box H) \geq \gamma_f(G) \cdot \gamma(H),
\end{equation}
where $\gamma_f(G)$ is the fractional domination 
number~\cite{kn:fisher} (see also~\cite{kn:bresar3}).  
The fractional domination number is, by linear programming duality 
equal to the fractional 2-packing number 
(see, eg,~\cite{kn:rubalcaba}). For strongly 
chordal graphs $\gamma_f(G)=\gamma(G)$~\cite{kn:scheinerman} 
and, therefore, Vizing's conjecture is true for strongly chordal 
graphs.  
Recently, more progress was made by Suen and Tarr~\cite{kn:suen}. 
They proved that 
\[\text{for all graphs $G$ and $H$} \quad 
\gamma(G \Box H) \geq \frac{1}{2}\cdot \gamma(G) \cdot \gamma(H)+ 
\frac{1}{2} \cdot \min \;\{\;\gamma(G),\;\gamma(H)\;\}.\] 

\bigskip 

Actually, in~\cite{kn:aharoni2} the authors show 
that 
for all graphs $G$ and $H$ 
\[\gamma(G \Box H) \geq \gamma^i(G) \cdot \gamma(H)
\quad\text{and}\quad \gamma^i(G \Box H) \geq \gamma^i(G) \cdot 
\gamma^i(H).\] 

\bigskip 

These result prompted us to investigate the computational 
complexity of $\gamma^i(G)$ for some classes of graphs.
We find that especially cographs, and related classes of graphs, 
deserve interest since they are completely decomposable 
by joins and unions and they are therefore susceptible 
to proofs by induction. 
As far as we know, 
the computational complexity 
of $\gamma(G \Box H)$ is still open for cographs.    

\bigskip 
 
Computing the domination number is NP-complete 
for chordal graphs~\cite{kn:bertossi,kn:booth}, 
and this implies the NP-completeness 
for the independence domination.  
A similar proof as in~\cite{kn:bertossi} shows that 
independence domination is NP-complete for bipartite graphs.    
It is NP-complete to decide whether $\gamma^i(G) \geq 2$ 
for weakly chordal graphs%
~\cite{kn:milanic}. 
The problem is polynomial for strongly chordal 
graphs~\cite{kn:farber2}.  

\section{Cographs}
\label{section cographs}

In this section we present our results for the class of cographs. 

\begin{definition}
A cograph is a graph without induced $P_4$. 
\end{definition}

\bigskip 

Cographs are the graphs $G$ that either have only one 
vertex, or for which either $G$ or $\Bar{G}$ is 
disconnected~\cite{kn:corneil}. Obviously, the class of 
graphs is hereditary in the induced subgraph order.  
It follows that a graph is a cograph if it is completely 
decomposable by joins and unions. We write $G=G_1 \oplus G_2$ 
when $G$ is the union of two smaller cographs $G_1$ and $G_2$ and 
we write $G= G_1 \otimes G_2$ when $G$ is the join of two 
smaller cographs $G_1$ and $G_2$.   

\begin{theorem}
\label{thm gamma cograph}
When $G$ is a cograph with at least two vertices then 
\[\gamma(G) = \begin{cases}
\min\;\{\;\gamma(G_1),\;\gamma(G_2),\;2\} & \text{if $G=G_1 \otimes G_2$,}\\
\gamma(G_1)+\gamma(G_2) & \text{if $G=G_1 \oplus G_2$.}
\end{cases}\]
\end{theorem}
\begin{proof}
When $G$ is the union of two graphs $G_1$ and $G_2$ then 
$\gamma(G)=\gamma(G_1)+\gamma(G_2)$, since no vertex of 
$G_1$ dominates a vertex of $G_2$ and vice versa. 

\medskip 

\noindent
Assume that $G=G_1 \otimes G_2$. Any pair of vertex $x \in V(G_1)$ and 
$y \in V(G_2)$ is a dominating set. When one of $G_1$ or $G_2$ has a 
universal vertex then that is a universal vertex for $G$. 
This proves the formula for the join. 
\qed\end{proof}

\bigskip 

\begin{remark} 
In~\cite{kn:barcalkin} a graph $G$ is called decomposable 
if its clique cover number is 
$\gamma(G)$, that is, if 
\[\chi(\Bar{G})=\gamma(G).\] 
The ``A-class'' is the collection 
of graphs that can be made decomposable by adding edges 
to it without changing the domination number. 
It is shown that graphs with domination number two, 
such as complete multi-partite graphs,  
belong to the A-class~\cite{kn:barcalkin}. 
According to~\cite{kn:barcalkin} 
Vizing's conjecture holds true for graphs in A-class 
(see also~\cite[Theorem~2]{kn:bresar2}). 
In~\cite{kn:aharoni2} 
the authors raise the interesting 
question whether chordal graphs are A-class 
graphs.  
\end{remark}

\bigskip 

\begin{theorem}
Let $G$ be a cograph. Then $\gamma^i(G)$ is the 
number of components of $G$. 
\end{theorem}
\begin{proof}
When $G$ has only one vertex then $\gamma^i(G)=1$.

\medskip 

\noindent 
Assume that $G=G_1 \otimes G_2$. 
Any maximal independent set is contained in $G_1$ or in 
$G_2$. To dominate it, one needs only one vertex, from 
the other constituent. 

\medskip 

\noindent
Assume that $G=G_1 \oplus G_2$. Then 
any maximal independent set is the union of a 
maximal independent set in $G_1$ and $G_2$. 
For the independence domination we have 
\[\gamma^i(G)=\gamma^i(G_1)+\gamma^i(G_2).\] 
By induction, $\gamma^i(G_j)$ is the number 
of components in $G_j$ for $j \in \{1,2\}$. 
\qed\end{proof}

\section{Distance-hereditary graphs}
\label{section DH}

Distance-hereditary graphs were introduced by Howorka as those 
graphs in which for every pair of nonadjacent vertices 
all the chordless paths that connect them have the same length%
~\cite{kn:howorka}. This class of graphs properly contains the 
class of cographs.  

\bigskip 

Distance-hereditary graphs $G$ have a decomposition tree $(T,f)$ 
which is described as follows 
(see~\cite{kn:oum} or, eg,~\cite{kn:kloks2}).  
Here, $T$ is a rooted binary tree and $f$ is a 
bijection from the vertices of $G$ to the leaves of $T$. 
Let $e$ be an edge of $T$ and let $W_e$ be the set of vertices 
that are mapped to the leaves in the subtree rooted at $e$. 
The ``twinset'' $Q_e \subseteq W_e$ is the set of vertices 
that have neighbors in $V \setminus W_e$. 

\bigskip 

Each internal node $p$ in the tree is labeled as $\otimes$ or $\oplus$. 
Let $e_1$ and $e_2$ be the two edges that connect $p$ with its 
children. Write $Q_1$ and $Q_2$ for the 
twinsets at $e_1$ and $e_2$. 
If the label of $p$ is $\otimes$ then all vertices 
of $Q_1$ are adjacent to all vertices of $Q_2$. 
If the label is $\oplus$ then no vertex of $Q_1$ is 
adjacent to any vertex of $Q_2$. 

\bigskip 

Let $e$ be the edge that connects $p$ with its parent. 
The twinset $Q_e$ is either 
\[Q_1 \quad\text{or}\quad Q_2 \quad\text{or}\quad 
Q_1 \cup Q_2 
\quad\text{or}\quad \es.\] 
The distance-hereditary graphs are exactly the graphs of 
rankwidth one. The decomposition tree above describes  
a rank-decomposition of width one. 

\bigskip 

\begin{theorem}
There exists an $O(n^3)$ algorithm that computes  
the independence domination number for distance-hereditary graphs. 
\end{theorem}
\begin{proof}
The decomposition tree can be computed in linear 
time~\cite{kn:damiand}. 
Let $e$ be an edge in the decomposition tree. Let $W_e$ be 
the set of vertices that are mapped to the leaves in the subtree 
and let $Q_e$ be the twinset, ie, the set of vertices in $W_e$ 
that have neighbors in $V \setminus W_e$.

\medskip 

\noindent
The algorithm computes a table for each edge $e$ 
in the decomposition tree. 
We write $H=G[W_e]$.  
For every pair of integers $a,g \in \{1,\dots,n\}$ the table stores 
a boolean value which is {\sc true} if there exists an 
independent set $A$ in $H$ with $|A|=a$ 
of which every vertex is dominated 
by a collection $D$ vertices in $H$ with $|D|=g$, 
except, possibly, some vertices 
in $A \cap Q_e$ (which are not dominated). 
The same table entry contains 
a boolean parameter which indicates whether 
there are vertices in $A \cap Q_e$ that 
are not dominated by the set $D$. A third boolean parameter  
indicates whether $D \cap Q_e$ is empty or not. Finally, a 
fourth boolean parameter stores whether some vertices of $D \cap Q_e$ 
dominate some vertices in $A \cap (W_e \setminus Q_e)$.   

\medskip 

\noindent
The information is conveniently stored in a symmetric 
$6 \times 6$ matrix. 
The rows and columns are partitioned according 
to the subsets 
\[A, \quad D, \quad 
A \cap Q_e, \quad D \cap Q_e, \quad A \cap (W_e \setminus Q_e) \quad 
\text{and}\quad D \cap (W_e\setminus Q_e). \] 
The diagonal entries indicate whether the subset is empty or not, 
and the off-diagonal entries indicate whether the subset of $D$ 
either completely dominates all the vertices,  
or partly dominates some of the vertices, 
or does not dominate any vertex of the subset of $A$. 

\medskip 

\noindent
We describe shortly some cases that illustrate how a table 
for an edge $e$ is computed. 
Consider a join operation at a node $p$. Let $e_1$ and 
$e_2$ be the two edges that connect $p$ with its children. 
An independent set $A$ in $G[W_e]$ can have vertices only in 
one of the two twinsets $Q_1$ and $Q_2$. 
Consider the case where $Q_e = Q_2$. When $Q_1$ has 
vertices in the independent set $A$ which are not dominated by 
vertices in $D_1$, then these 
vertices have to be dominated by a vertex from 
$Q_2$. In case of a join operation, any (single) 
vertex of $Q_2$ can do  
the job. When a dominating set $D_2$ has vertices 
in $Q_2$ then this vertex dominates $A \cap Q_1$. 
Otherwise, a new vertex of $Q_2$ 
needs to be added to the dominating set.  

\medskip 

\noindent
It is easy to check that a table as described above 
can be computed for each edge $e$  
from similar tables stored at the two children of $e$. 
For brevity we 
omit further details. The independence number 
can be read from the table at the root. 
\qed\end{proof}

\bigskip 

\begin{remark}
It is easy to see that this generalizes to graphs 
of bounded rankwidth. As above, let $(T,f)$ be a decomposition tree. 
Each edge $e$ of $T$ partitions the vertices of $G$ into two sets. 
The cutmatrix of $e$ is the submatrix of the adjacency matrix 
that has its rows indexed by the vertices in one part of the partition 
and its columns indexed by the vertices in the other part of the 
partition. A graph has rankwidth $k$ if the rank over $GF[2]$ of every 
cutmatrix is at most $k$. For example, when 
$G$ is distance hereditary, then every edge in the 
decomposition tree has a cutmatrix with a shape 
$\bigl( \begin{smallmatrix} J & 0 \\ 0 & 0 \end{smallmatrix}\bigr )$  
where $J$ is the all-ones matrix. Thus every cutmatrix has rank one.  
When a graph has bounded rankwidth then 
the twinset $Q_e$ of every edge $e$ has a partition into a bounded number 
of subsets. The vertices within each subset have the same 
neighbors in $V \setminus W_e$~\cite{kn:kloks5}. A rank-decomposition 
tree of bounded width can be obtained in $O(n^3)$ time~\cite{kn:oum}.     
\end{remark}

\section{Permutation graphs}
\label{section perm}

Another class of graphs that contains the cographs is the 
class of permutation graphs~\cite{kn:golumbic}. 

\bigskip 

A permutation diagram 
consists of two horizontal lines in the plane and a collection 
of $n$ line segments, each connecting a point on the 
topline with a point on the bottom line. A graph is a 
permutation graph if it is the intersection graph of 
the line segments in a permutation diagram. 

\bigskip 

In their paper Baker, Fishburn and Roberts characterize permutation 
graphs as follows~\cite{kn:baker2}. (See also~\cite{kn:dushnik}; 
in this paper the authors characterize permutation graphs as 
interval containment graphs). 

\begin{theorem}
A graph $G$ is a permutation graph if and only if 
$G$ and $\Bar{G}$ are comparability graphs. 
\end{theorem}

Assume that $G$ and $\Bar{G}$ are comparability graphs. 
Let $F_1$ and $F_2$ be transitive orientations of $G$ and 
$\Bar{G}$. A permutation diagram for $G$ is obtained by ordering 
the vertices on the topline by the total order $F_1 \cup F_2$ 
and on the bottom line by the total order $F_1^{-1} \cup F_2$. 
Permutation graphs can be recognized in 
linear time. The algorithm can be used to produce a 
permutation diagram in linear time~\cite{kn:tedder}. 

\bigskip 

Consider a permutation diagram for a permutation 
graph $G$. An independent set $M$ in $G$ corresponds with a 
collection of parallel line segments. The line segments of 
vertices in $M$ are, therefore, linearly ordered,  
say left to right. 
 
\begin{definition}
Consider a permutation diagram. An independent set $M$ 
ends in $x$ if the line segment of $x$ is the right-most 
line segment of vertices in $M$. 
\end{definition}

\begin{definition}
For $x \in V$ and $k \in \mathbb{N}$, let $\mathcal{M}(x;k)$ be the 
collection of independent sets $M$ that end in $x$ and 
for which $\gamma(M)=k$. 
\end{definition}

\begin{definition}
Let $\Gamma(x;k)$ be the collection 
of minimum dominating sets for independent sets $M$ that end 
in $x$ with $\gamma(M)=k$. 
\end{definition}

\bigskip 

The line segments of the neighbors of a vertex $x$ 
are crossing the line segment of $x$. We say that $z$ 
is a rightmost neighbor of $x$ satisfying a certain 
condition, if the endpoint 
of $z$ on either the topline or the bottom line is rightmost 
among all neighbors of $x$ that satisfy the condition. 
Here, we allow that $z=x$. 

\bigskip 

Let $x \in V$ and let $z \in N[x]$. Define 
\begin{equation}
\label{eqn1}
\gamma_x(z)=\{\;k\;|\; \text{$z$ is a right-most neighbor of $x$ 
and $z \in \Gamma$ for some $\Gamma \in \Gamma(x;k)$}\;\}
\end{equation}

\bigskip 

\begin{lemma}
Let $G$ be a permutation graph and consider a permutation diagram 
for $G$. Then 
\begin{equation}
\label{eqn2}
\gamma^i(G)=\max \;\{\;k\;|\; k \in \gamma_x(z) \quad x \in V \quad 
z \in N[x]\;\}.
\end{equation}
\end{lemma}
\begin{proof}
Consider an independent set $M \subseteq V$ for which 
\[\gamma(M)=\gamma^i(G).\]  
Assume that $M$ ends in $x$. Any set $\Gamma$ that dominates 
$M$ has a vertex $z \in N[x] \cap \Gamma$.   
Let $z$ be a right-most neighbor of $x$ which is in a 
dominating set $\Gamma$ for $M$ with $|\Gamma|=\gamma(M)$. Then 
\begin{equation}
\label{eqn3}
\gamma^i(G) =\gamma(M) \in \gamma_x(z).
\end{equation}
This proves the lemma. 
\qed\end{proof}

\bigskip 

\begin{theorem}
There exists a polynomial-time algorithm that computes 
$\gamma^i(G)$ for permutation graphs. 
\end{theorem}
\begin{proof}
We describe the algorithm to compute $\gamma_x(z)$. 
We assume that for every non-neighbor $y$ of $x$ 
that is to the left of 
$x$, the sets $\gamma_y(z^{\prime})$ for $z^{\prime} \in N[y]$ 
have been computed. 

\medskip 

\noindent
Consider an independent set $M \in \mathcal{M}(x;k)$. 
Let $z \in N[x]$ be a rightmost 
neighbor of $x$ such that there is a dominating set 
$\Gamma \in \Gamma(x;k)$ with $z \in \Gamma$. 
Let $y \in M$ lie immediately to the left of $x$. 
When $z \in N(y)$ then $z$ must be a rightmost neighbor 
of $y$. In that case 
\begin{equation}
\label{eqn6}
k \in \gamma_x(z) \quad\Leftrightarrow\quad k \in \gamma_y(z).
\end{equation}

\medskip 

\noindent
Now assume that $z \notin N(y)$. Then $z$ dominates only one 
vertex of $M$, namely $x$. In that case $z$ must be a right-most 
neighbor of $x$ which is not in $N(y)$ and, if that is the case,  
\begin{equation}
\label{eqn7}
k \in \gamma_x(z) \quad \Leftrightarrow \quad 
\exists_{z^{\prime} \in N[y]\setminus N(x)} \; k-1 \in \gamma_y(z^{\prime}).
\end{equation}

\medskip 

\noindent
This proves the theorem. 
\qed\end{proof}
 
\section{Bounded treewidth}
\label{section tw}

Graphs of bounded treewidth were introduced by Halin~\cite{kn:halin}. 
They play a major role in the research on graph minors~\cite{kn:robertson}. 
Problems that can be formulated in monadic second-order 
logic can be solved in linear time for graphs of 
bounded treewidth.  
Graphs of bounded treewidth can be recognized in 
linear time~\cite{kn:bodlaender,kn:kloks3}. Actually, 
bounded treewidth itself can be formulated in monadic 
second-order logic via a finite collection of 
forbidden minors~\cite{kn:courcelle}. 

\begin{definition}
Let $k \in \mathbb{N}$. A graph $G$ has treewidth at most $k$ 
if $G$ is a subgraph of a chordal graph $H$ with $\omega(H) \leq k+1$. 
\end{definition}

\begin{theorem}
Let $k \in \mathbb{N}$. 
There exists an $O(n^3)$ algorithm to compute 
$\gamma^i(G)$ when the treewidth of $G$ is at most $k$. 
\end{theorem}
\begin{proof}
Consider a tree-decomposition for $G$ with bags of size 
at most $k+1$~\cite{kn:kloks3,kn:kloks2}. 
Consider a subtree rooted at a node $i$. 
Denote the bag at node $i$ by $S_i$. Denote the 
subgraph of $G$ induced by the vertices that appear in bags in the 
subtree rooted at $i$ by $G_i$. 
We use a technique similar to the one used in, 
eg,~\cite{kn:telle,kn:telle3}.  

\medskip 

\noindent
For all the subsets $A \subseteq S_i$, 
and for all pairs of integers $p$ and $q$, let $b(p,q,A)$ 
denote a boolean value which is true if there 
exists an independent set $M$ in $G_i$ with $p$ vertices 
with $M \cap S_i=A$. The vertices of $A$ have a status, which is 
either white or gray. The white vertices of $A$ are dominated 
by a set of $q$ vertices in $G_i$ and the gray vertices are not 
dominated by vertices in $G_i$. 

\medskip

\noindent
It is easy to see that the boolean values can be computed in 
$O(n^2)$ time by dynamic programming for each node in the 
decomposition tree. 
\qed\end{proof}
 
\section{An exact exponential algorithm}
\label{section exponential}

In this section we describe an exact, exponential 
algorithm to compute the independence domination 
number~\cite{kn:fomin,kn:kloks2}. 

\bigskip 

\begin{theorem}
There exists an $O^{\ast}(1.7972^n)$ algorithm to 
compute the independence domination number. 
\end{theorem}
\begin{proof}
Moon and Moser proved that a graph with 
$n$ vertices has at most $3^{n/3}$ maximal independent 
sets~\cite{kn:moon}. 
Tsukiyama et al. showed that all the independent sets 
can be listed with polynomial delay~\cite{kn:tsukiyama}. 

\medskip 

\noindent 
First assume that there is a maximal independent set 
with at most $\beta \cdot n$ vertices. We determine the 
constant $\beta$ later. 
Then $\gamma^i(G) \leq \gamma(G) \leq \beta \cdot n$. 

\medskip 

\noindent
For each maximal independent set $M$ of size at most $\beta \cdot n$, 
we find the smallest set that dominates it as follows. 
Remove all edges except those that connect 
$M$ and $V \setminus M$. Assume that every vertex of $V \setminus M$ 
has at most two neighbors in $M$. Then 
we can easily find $\gamma(M)$ in polynomial time via 
maximum matching. To see that, construct a graph $H$ on the 
vertices of $M$ where two vertices are adjacent if they have a common 
neighbor in $V \setminus M$. Let 
$W$ be the set of vertices in $M$ that are endpoints of 
edges in a maximum matching. Let $\nu(H)$ be the cardinality 
of a maximum matching in $H$. 
Then a solution is given by 
\begin{equation}
\label{eqn matching}
\gamma(M)=\nu(H)+ |M \setminus W|. 
\end{equation}

\medskip 

\noindent 
Otherwise, when at least some vertex 
of $V \setminus M$ has at least three 
neighbors in $M$, choose a vertex $x$ of maximal degree at least three 
in $V \setminus M$ and branch as follows. In one branch 
the algorithm removes $x$ and all its neighbors. In the other 
branch only the vertex $x$ is removed. This gives a recurrence 
relation 
\[T(n) \leq T(n-1)+T(n-4).\] 
Since the depth of the search tree 
is bounded by $\beta \cdot n$, this part of the algorithm can be solved 
in $O^{\ast}(1.3803^{\beta \cdot n})$. 

\medskip 

\noindent 
Assume that every maximal independent set has cardinality 
at least $\beta \cdot n$. In that case, we try all subsets 
of $V \setminus M$. The optimal value for $\beta$ follows 
from the equation 
\[1.3803^{\beta} = 2^{1-\beta} \quad\Rightarrow\quad 
\beta=0.6827.\]
For the timebound we find that it is polynomially equivalent to 
\[3^{n/3} \cdot 2^{(1-\beta) n} =1.7972^n.\]
\qed\end{proof}
  
\section{A PTAS for planar graphs}

In this section we show that there is a polynomial-time 
approximation scheme for planar graphs. 
We use the well-known technique of Baker~\cite{kn:baker}. 

\bigskip 

Consider a plane embedding of a planar graph $G$. Partition 
the vertices of $G$ into layers $L_1, \dots$ as follows. 
The outerface are the vertices of $L_1$. Remove the vertices 
of $L_1$. Then the new outerface are the vertices of $L_2$. 
Continue this process until all vertices are in some layer. 

\bigskip 

If there are only $k$ layers then the graph is called $k$-outerplanar. 

\begin{lemma}[\cite{kn:bodlaender2}]
The treewidth of $k$-outerplanar graphs is at most $3k-1$. 
\end{lemma}

\begin{theorem}
Let $G$ be a planar graph. 
For every $\epsilon > 0$ there exists a linear-time 
algorithm that computes an independence dominating 
set of cardinality at least 
\[(1-\epsilon)\cdot \gamma^i(G).\]  
\end{theorem}
\begin{proof}
Let $k \in \mathbb{N}$. Let $\ell \in \{1,\dots,k\}$ and consider 
removing layers 
\[L_{\ell}, L_{\ell+k}, L_{\ell+2k}, \dots.\] 
Let $G(\ell,k)$ be the remaining graph. Then every component 
of $G$ has at most $k$ layers, and so $G(\ell,k)$ 
has treewidth at most $3k-1$. Using the algorithm of 
Section~\ref{section tw} we can compute the independence domination 
numbers of $G(\ell,k)$, for $\ell \in \{1,\dots,k\}$. 

\medskip 

\noindent
Let $M$ be an independent set in $G$ with $\gamma(M)=\gamma^i(G)$. 
If we sum over $\ell \in \{1,\dots,k\}$, the vertices of $M$ 
are counted $k-1$ times. Each $\gamma^i(G(\ell,k))$ is 
at least as big as the dominating set that is needed to 
dominate the remaining vertices of $M$. Therefore, the sum 
over $\gamma^i(G(\ell,k))$ is at least $(k-1)\cdot \gamma^i(G)$. 
Therefore, if we take the maximum of $\gamma^iG(\ell,k))$ 
over $\ell \in \{1,\dots,k\}$ we find an approximation 
of size at least $(1-\frac{1}{k})\cdot\gamma^i(G)$. 
\qed\end{proof}
   
\section{Concluding remarks}

One of our motivations to look into the 
independence domination number for classes of perfect graphs 
is the domination problem for edge-clique graphs of 
cographs. The main reason to look into this 
are the recent complexity results on 
edge-clique covers~\cite{kn:cygan,kn:impagliazzo}.

Let $G=(V,E)$ be a graph. The edge-clique graph 
$K_e(G)$ is the graph which has $E$ as its vertices and 
in which two elements of $E$ are adjacent when they 
are contained in a clique of 
$G$~\cite{kn:albertson,kn:cerioli,kn:gregory,%
kn:raychaudhuri,kn:raychaudhuri2}. 

\medskip 

Let $G$ and $H$ be two 
graphs. 
The strong product $G \boxtimes H$ is the subgraph 
of $K_e(G \otimes H)$ induced by the edges that have one 
endpoint in $G$ and the other in $H$. In other words, 
the vertices of $G \boxtimes H$ are pairs $(g,h)$ with 
$g \in V(G)$ and $h \in V(H)$. Two vertices $(g_1,h_1)$ and 
$(g_2,h_2)$ are adjacent when $g_1 \in N[g_2]$ and 
$h_1 \in N[h_2]$. It is 
well-known~\cite{kn:korner,kn:lovasz,kn:ramirez,kn:shannon} 
that, when $G$ and $H$ are perfect,  
\[\alpha(G \boxtimes H) = \alpha(G) \cdot \alpha(H).\] 
Notice however that $G \boxtimes H$ itself is not 
necessarily perfect. 
For example $C_4 \boxtimes C_4$ contains an induced $C_5$. 
The determination of 
$\alpha(G \boxtimes G)$ is very hard when $G$ is 
not perfect. Lov\'asz proved that 
$\alpha(C_5 \boxtimes C_5) = \sqrt{5}$ but, as far as we know, 
$\alpha(C_7 \boxtimes C_7)$ is open~\cite{kn:lovasz}. 

\medskip 

The independence number of the strong product has been investigated 
a lot due to its applications in data compression and coding theory. 
Very little is known about the (independent) domination number 
of strong products, although some investigations 
were made 
in~\cite{kn:bisbee,kn:bresar,kn:domke,%
kn:farber3,kn:farber2,kn:farber,kn:fisher,%
kn:hare,kn:hell,kn:kloks4,kn:ma,kn:mceliece,kn:park,%
kn:rubalcaba,kn:scheinerman}.         

\medskip 

As far as we know, the domination number 
for the edge-clique graphs 
of complete multipartite graphs is open. For simplicity, we call this 
the edge-domination number.\footnote{One should be cautious because 
this terminology is also used for 
a different concept.} A minimum edge-domination set is 
not necessarily realized by 
the complete bipartite subgraph induced by the two 
smallest color classes. For example, $K(2,2,2)$ has edge-domination 
number three while the complete bipartite 
$K(2,2)$ has four edges. The edge-clique cover for complete multipartite 
graphs seems to be a very hard problem~\cite{kn:kloks,kn:ma,kn:park}.

\end{document}